\documentclass[11pt, a4paper]{article}

\usepackage[margin=2.5cm]{geometry} 
\usepackage{setspace}
\usepackage[utf8]{inputenc}
\usepackage[T1]{fontenc}

\usepackage{amsmath, amsthm, amssymb}
\usepackage{mathrsfs}
\usepackage{mathtools}
\mathtoolsset{showonlyrefs=true}

\usepackage{titlesec}
\titleformat*{\section}{\large\bfseries}
\titleformat*{\subsection}{\normalsize\bfseries}

\titleformat{\paragraph}[runin]
  {\normalfont\slshape}
  {\theparagraph}
  {1em}
  {}
  [. ]

\usepackage{caption}
\usepackage{subcaption}
\usepackage{booktabs}
\usepackage{threeparttable}
\usepackage{siunitx}
\usepackage{multirow}
\usepackage{tabularx}
\usepackage{etoolbox}
\robustify\bfseries

\usepackage[round,authoryear]{natbib}

\usepackage[colorlinks=true,allcolors=blue]{hyperref}

\usepackage{lineno}

\usepackage{authblk}

\theoremstyle{plain}
\newtheorem{proposition}{Proposition}
\theoremstyle{definition}
\newtheorem*{definition}{Definition}
\newtheorem{lemma}{Lemma}

\title{\large \textbf{
Computing Endogenous Transformations in Processing Networks: \\A Dynamic Calibration Approach
}}
\author[a]{Satoshi Nakano\thanks{Email: nakano@n-fukushi.ac.jp}}
\author[b]{Kazuhiko Nishimura\thanks{Email: nishimura@lets.chukyo-u.ac.jp (Corresponding Author)}}
\affil[a]{Faculty of Economics, Nihon Fukushi University, Tokai 477-0031 Japan}
\affil[b]{Institute of Economics, Chukyo University, Nagoya 466-8666 Japan}

\date{\small \today} 

\begin{document}

\maketitle


\begin{abstract}
Understanding how production networks endogenously transform requires a parametric model with non-neutral substitution elasticities. While the Cascaded CES (CCES) production function provides a rigorous framework, dynamically calibrating its structural parameters from time-series data constitutes a highly non-convex inverse optimization problem that typically falls prey to the curse of dimensionality. To overcome this limitation, we propose a structure-exploiting hybrid algorithm that leverages the physical upstreamness topology of the network. Applying this framework to U.S. time-series data, we calibrate an economy-wide multisector processing network and provide robust evidence that the empirical supply chain is structurally inelastic. Through Monte Carlo simulations, we reveal that this structural inelasticity acts as a macroeconomic \textit{impedance}; it endogenously transforms symmetric microeconomic volatility into heavy-tailed asymmetric risks, causing a statistically significant \textit{mean sinking} effect that structurally depresses expected aggregate productivity. 
\end{abstract}


\section{Introduction}
\label{sec:introduction}
In the modern global economy, predicting how production networks endogenously transform in response to exogenous shocks---and accurately measuring the structural dynamics underlying these transformations---remains a central challenge in macroeconomic analysis. To effectively evaluate these complex dynamics, it is essential to distinguish between a static \textit{supply chain} and a general equilibrium \textit{processing network}. We define the former as the empirically observed snapshot of adopted technologies and input-output linkages at a specific point in time. In contrast, we define a \textit{processing network} as the underlying theoretical engineering system that encompasses the complete hierarchical structure and all potential technological substitutions available to the economy. Ultimately, analyzing and predicting the endogenous transformation of a production network necessitates the rigorous parametric modeling of such an underlying engineering system.

The epistemological foundation of this approach traces back to the pioneering vision of \cite{Chenery1949}. By breaking down production into elemental physical processes described by engineering variables, Chenery demonstrated that constructing ``engineering production functions'' can overcome the limitations of conventional models that treat technology as an abstract statistical black box. Building upon this philosophy, an advanced processing network is capable of describing how an exogenous cost increase in a specific primary factor propagates through relative prices, prompting sectors to endogenously substitute inputs and ultimately reconstruct the observable macroeconomic snapshot.

Traditionally, the structural analysis of economic networks has relied heavily on the costly compilation of time-series Input-Output (IO) tables \citep{Jorgenson1987, Timmer2015}. While non-parametric index approaches, such as the T\"ornqvist index, serve as excellent ex-post accounting tools---effectively acting as historical dictionaries of structural changes---they inherently lack predictive capabilities. They can record how a network's cost shares transformed, but they cannot predict how it will reconstruct itself when faced with unprecedented counterfactual shocks. On the other end of the spectrum, standard computable general equilibrium (CGE) models designed for ex-ante policy simulations typically calibrate their baseline to a single year's data and rely on static, literature-based substitution elasticities \citep{Dawkins2001, Dixon2013}. Consequently, existing frameworks struggle to dynamically disentangle genuine technological shifts from price-induced input substitution within the intricate, multi-layered linkages of the full production network.

This limitation is particularly problematic because recent advancements in macroeconomics have established that network granularity and heterogeneity are essential for understanding aggregate fluctuations \citep{Gabaix2011, Acemoglu2012}. While recent theoretical literature has demonstrated that network non-linearities can generate synergies among multiple productivity shocks \citep{KimNakanoNishimura2017, NakanoNishimura2026} and induce aggregate tail risks from symmetric shocks \citep{BaqaeeFarhi2019_2}, a critical empirical gap remains.
Determining whether a real-world economy is structurally \textit{inelastic} enough to trigger these vulnerabilities---and quantifying the resulting macroeconomic drag---requires dynamically identifying these heterogeneous substitution elasticities from time-series data. If the substitution elasticity is strictly unity (a neutral Cobb-Douglas assumption), the cost shares remain mathematically invariant, meaning no endogenous structural transformation occurs. Thus, moving beyond static, neutral frameworks is strictly necessary to capture true macroeconomic dynamics \citep{LiuTsyvinski2023, Kopytov2024}.

To rigorously capture these endogenous reorganizations, \cite{NakanoNishimura2018, NakanoNishimura2021} introduced Cascaded CES (CCES) production functions. This approach formulates \textit{nonneutroelastic processing networks}---a comprehensive framework where heterogeneous, non-neutral substitution elasticities govern structural transformations at each stage of production. Yet, calibrating a large-scale CCES network against time-series data constitutes a highly non-linear and non-convex problem. Standard unconstrained econometric regressions frequently yield negative elasticities, severely violating the global concavity required for microeconomic regularity. Conversely, when applying generic gradient-based solvers for constrained optimization, the strong parameter interferences inherent in the nested CCES hierarchy trap the algorithm in active-set boundaries, rendering economywide estimations computationally intractable.

To overcome this computational barrier and advance empirical structural analysis, this paper proposes a ``Dynamic Calibration'' (DC) approach. Our core methodological contribution is the development of a deterministic \textit{structure-exploiting algorithm}. While mathematically related to coordinate descent methods, our approach avoids the pitfalls of generic solvers by explicitly leveraging the topological upstreamness (the physical hierarchy) of the supply chain. By alternating between a vertical cascade-sequential descent that neutralizes nested ill-conditioning and a horizontal block coordinate descent that resolves general equilibrium interference, our algorithm effectively breaks the curse of dimensionality.

Applying this framework to United States time-series data, we transition the analysis of production networks from ex-post descriptive accounting to predictive structural simulations. First, we demonstrate that our endogenously calibrated parameters successfully recover the conventional data-driven T\"ornqvist TFP index, providing a robust methodology to capture complex technological dynamics without violating microeconomic concavity. 

Second, we utilize the calibrated processing network to uncover the origins of systemic vulnerabilities. Through Monte Carlo simulations, we reveal that the empirically calibrated CCES network intrinsically translates standard symmetric micro-shocks into heavy-tailed, asymmetric macroeconomic losses. Ultimately, this framework provides empirical evidence that the structural inelasticity of the network generates a macroeconomic \textit{impedance}---a systemic \textit{mean sinking} effect that permanently depresses aggregate output.

The remainder of this paper is organized as follows. Section \ref{sec:processing_networks} formally defines the economywide processing networks and the theoretical structure of the CCES model. Section \ref{sec:methodology} details the Dynamic Calibration problem and introduces our structure-exploiting algorithm. Section \ref{sec:results} presents the empirical application to the US economy, validates the calibration through structural TFP comparisons, and evaluates macroeconomic tail risks through the lens of mean sinking. Finally, Section \ref{sec:conclusion} concludes the paper.

\section{Processing Networks Framework}
\label{sec:processing_networks}

Processing networks provide an analytical bridge between the engineering reality of hierarchical supply chains and the economic depiction of market transactions. To translate this structural topology into an economically rigorous general equilibrium model, we must define an aggregate cost function that explicitly honors the sequential nature of material flows (see, Figure \ref{fig:cces_structure}).
\begin{figure}[t!]
\centering
\includegraphics[width=0.7\textwidth]{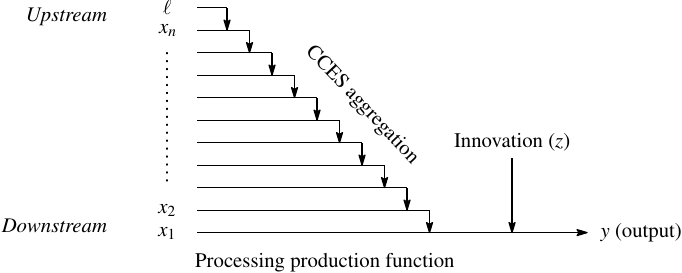} 
\vspace{0.0em}
\caption{A schematic representation of the processing production function, which constitutes the underlying engineering structure of the Processing Networks.}
\label{fig:cces_structure}
\end{figure}


\subsection{Unit Cost Functions and Cost Share Matrix}
\label{subsec:cces_cost_shares}

To capture the sequential and hierarchical nature of cost propagation in processing networks, we employ the Cascaded CES (CCES) unit cost function under the assumption of constant returns to scale. For a given industry $j \in \{1, \dots, n\}$, let $i \in \{1, \dots, n\}$ denote the hierarchical level of intermediate inputs, ordered from the most downstream (final processing) input to the most upstream (raw extraction) input. The recursive structure of the CCES unit cost function is formalized as follows:\footnote{The underlying CCES production function that corresponds to this dual cost structure is detailed in Eq. \eqref{cces_pf}.}
\begin{align}
\pi_{i-1} = \left( \lambda_i (p_i)^{1 - \sigma_i} + (1 - \lambda_i) (\pi_i)^{1 - \sigma_i} \right)^{\frac{1}{1 - \sigma_i}}, \quad i = 1, \dots, n,
\label{eq:cces_recursive}
\end{align}
where $\lambda_i \in [0, 1]$ represents the distributional share parameter, $\sigma_i \ge 0$ dictates the elasticity of substitution at level $i$, and $p_i$ is the market price of the $i$-th factor input. The state variable $\pi_0$ represents the final aggregate unit cost at the bottommost level ($i=1$), which equates to the theoretical output price of industry $j$ before productivity adjustments. Conversely, $\pi_n = w$ represents the price of the primary factor (value-added) introduced at the uppermost foundational level ($i = n$).

By recursively unrolling Eq. \eqref{eq:cces_recursive}, the composite CCES unit cost function for industry $j$ can be expressed as a non-linear function of all input prices and the primary factor price:
\begin{align}
{\pi}_{0} = {C} \left( p_1, \dots, p_n, w \mid \lambda_1, \dots, \lambda_n ;\, \sigma_1, \dots, \sigma_n \right),
\qquad
p = z^{-1}{\pi}_{0}, 
\label{eq:ucf_compact}
\end{align}
where $z$ denotes the Hicks-neutral total factor productivity (TFP) level, and $p$ is the actual output price.

To evaluate the cost share of the $k$-th factor input, $s_k$, we apply Shephard's lemma. The explicitly hierarchical structure of the CCES function allows us to utilize the chain rule of differentiation. The partial derivatives of the recursive function are given by:
\begin{align}
\frac{\partial \pi_{i-1}}{\partial p_i} = \lambda_i \left( \frac{\pi_{i-1}}{p_i} \right)^{\sigma_i},
&&
\frac{\partial \pi_{i-1}}{\partial \pi_i} = (1 - \lambda_i) \left( \frac{\pi_{i-1}}{\pi_i} \right)^{\sigma_i}.
\end{align}
Consequently, the cost share $s_k$ for $k = 1, \dots, n$ is derived as:
\begin{align}
s_k 
= \frac{p_k}{\pi_0} \frac{\partial \pi_0}{\partial p_k} 
&= \frac{p_k}{\pi_0} \frac{\partial \pi_0}{\partial \pi_1} \cdots \frac{\partial \pi_{k-2}}{\partial \pi_{k-1}} \frac{\partial \pi_{k-1}}{\partial p_k} \nonumber \\
&= \frac{p_k}{\pi_0}  (1-\lambda_1)\left( \frac{\pi_{0} }{\pi_{1}} \right)^{\sigma_1} \cdots (1-\lambda_{k-1}) \left(\frac{\pi_{k-2} }{\pi_{k-1}} \right)^{\sigma_{k-1}} \lambda_{k}\left( \frac{\pi_{k-1}}{p_{k}} \right)^{\sigma_k}.
\label{eq:cost_share_k}
\end{align}
Equation \eqref{eq:cost_share_k} clearly illustrates how a price shock at level $k$ non-linearly cascades down the processing network to affect the final cost share. Because $s_k$ implicitly depends on the final unit cost $\pi_0$, the $i$-th factor cost share for industry $j$ is a function of the entire price vector and all structural parameters. Collecting these shares across all intermediate inputs $i = 1, \dots, n$ and value-added elements for industries $j = 1, \dots, n$, we define the economywide cost share matrix $\mathbf{S} \in \mathbb{R}^{(n+1) \times n}$:
\begin{align}
\mathbf{S}(\mathbf{p}, \mathbf{w} \mid {\Lambda}, {\Sigma}) = \left[ s_{ij} \left( p_1, \dots, p_n; w_j \mid \lambda_{1j}, \dots, \lambda_{nj}; \sigma_{1j}, \dots, \sigma_{nj} \right) \right],
\end{align}
where $\mathbf{p} = (p_1, \dots, p_n)^\prime$ is the intermediate price vector, $\mathbf{w} = (w_1, \dots, w_n)$ is the primary factor price vector, and ${\Lambda}$ and ${\Sigma}$ collect the corresponding share parameters and elasticities for all sectors. 

Correspondingly, the general equilibrium price is defined as the fixed point of the economywide system of price equations \eqref{eq:ucf_compact}, compactly described as:
\begin{align}
\mathbf{p} = \left< \mathbf{z} \right>^{-1} \mathcal{C}(\mathbf{p}, \mathbf{w} \mid {\Lambda}, {\Sigma}),
\label{economy_price}
\end{align}
where $\left< \mathbf{z} \right> = \text{diag}(z_1, \dots, z_n)$ is the diagonal matrix of sectoral productivities, and $\mathcal{C}$ maps the vector functions of unit costs.

A fundamental microeconomic requirement for any well-behaved unit cost function is that the sum of all factor cost shares (including the primary factor) equals unity. To formally prove this for our CCES structure, we define auxiliary variables:
\begin{align}
\mu_i = \lambda_i \left( \frac{\pi_{i-1}}{p_i} \right)^{\sigma_i - 1}, 
&&
\nu_i = (1 - \lambda_i) \left( \frac{\pi_{i-1}}{\pi_i} \right)^{\sigma_i - 1}.
\label{sector_price}
\end{align}
From Eq. \eqref{eq:cces_recursive}, it trivially holds that $\mu_i + \nu_i = 1$ for all $i = 1, \dots, n$. We hypothesize that the cumulative residual share follows the relationship:
\begin{align}
\prod_{i=1}^k \nu_i = 1 - \sum_{i=1}^k s_i.
\label{eq:induction_hypothesis}
\end{align} 
For $k=1$, applying Shephard's lemma yields $s_1 = \mu_1 = 1 - \nu_1$, satisfying the base case. Assuming Eq. \eqref{eq:induction_hypothesis} holds for $k$, we use Eq. \eqref{eq:cost_share_k} to observe that $s_{k+1} = \mu_{k+1} \prod_{i=1}^k \nu_i$. Thus:
\begin{align}
\nu_{k+1} = 1 - \mu_{k+1} = 1 - \frac{s_{k+1}}{\prod_{i=1}^k \nu_i} = \frac{1 - \sum_{i=1}^{k+1} s_i}{1 - \sum_{i=1}^k s_i}.
\end{align}
It subsequently follows that:
\begin{align}
\prod_{i=1}^{k+1} \nu_i 
= \nu_{k+1} \prod_{i=1}^k \nu_i 
= \frac{1 - \sum_{i=1}^{k+1} s_i}{1 - \sum_{i=1}^k s_i} 
\left( 1 - \sum_{i=1}^k s_i \right)
= 1 - \sum_{i=1}^{k+1} s_i.
\end{align}
By mathematical induction, Eq. \eqref{eq:induction_hypothesis} is valid for all $k$. Furthermore, the cost share of the uppermost primary factor (value-added), denoted $s_{\ell}$, is given by applying Shephard's lemma with respect to $w = \pi_n$:
\begin{align}
s_{\ell} 
= \frac{\pi_n}{\pi_0} \frac{\partial \pi_0}{\partial \pi_n} 
= \frac{\pi_n}{\pi_0} \frac{\partial \pi_0}{\partial \pi_1} \cdots \frac{\partial \pi_{n-1}}{\partial \pi_n} 
= \prod_{i=1}^n (1 - \lambda_i) \left( \frac{\pi_i}{\pi_{i-1}} \right)^{1-\sigma_i} 
= \prod_{i=1}^n \nu_i.
\label{eq:primary_share}
\end{align}
Combining Eq. \eqref{eq:induction_hypothesis} (evaluated at $k=n$) and Eq. \eqref{eq:primary_share} 
 confirms the requisite adding-up property:
\begin{align}
s_{\ell} = 1 - \sum_{i=1}^n s_i 
\end{align}


\subsection{Baseline Share Parameters}
\label{subsec:reference_state}

A critical challenge in parameterizing the CCES function is determining the sequence of the hierarchical levels $i = 1, \dots, n$ and mapping the baseline data to the unobservable share parameters ${\Lambda}$. In this framework, the hierarchical cascade is not arbitrarily assigned; it is objectively determined by a Leontief-based upstreamness metric, as formally defined in Eq. \eqref{eq:upstreamness} within Appendix \ref{AppendixB}. This metric ensures that intermediate inputs are strictly ordered from the most downstream (direct) inputs to the most upstream (indirect) inputs, naturally culminating in the foundational primary factor.

Once the cascade sequence is established, we calibrate the distributional parameters ${\Lambda}$ using the reference state observed in the baseline period ($t=0$). At this reference state, we normalize all observable market prices to unity (i.e., $p_i = 1$ and $w = 1$ for all $i$). According to the recursive structure outlined in Eq. \eqref{eq:cces_recursive}, this normalization immediately forces all intermediary composite unit costs to be unity as well ($\pi_i = 1$ for all $i$).

Let $\mathbf{A} = [a_{ij}]$ denote the $n \times n$ reference input-output coefficient matrix, where $a_k$ (dropping the industry index $j$ for brevity) represents the observed cost share of the $k$-th intermediate factor at the baseline. Evaluating Eq. \eqref{eq:cost_share_k} and the inductive property from Eq. \eqref{eq:induction_hypothesis} at the reference state ($\pi_i = 1$) yields:
\begin{align}
a_k = \lambda_k \prod_{i=1}^{k-1} (1 - \lambda_i), \qquad \text{and} \qquad \prod_{i=1}^k (1 - \lambda_i ) = 1 - \sum_{i=1}^k a_i.
\end{align}
This system allows us to uniquely identify the unobservable structural share parameter $\lambda_k$ solely from the observable baseline cost shares:
\begin{align}
\lambda_k = \frac{a_k}{\prod_{i=1}^{k-1} (1 - \lambda_i )} = \frac{a_k}{1 - \sum_{i=1}^{k-1} a_i },
\label{eq:lambda_determination}
\end{align}
where the primary factor share is given by the residual $a_{n+1} = 1 - \sum_{i=1}^{n} a_{i}$.

Equation \eqref{eq:lambda_determination} demonstrates that the entire share parameter matrix $\Lambda$ is strictly and deterministically identified by the baseline input-output matrix $\mathbf{A}$. Consequently, the time-varying cost share matrix is fundamentally anchored by $\mathbf{A}$ and dynamically driven solely by the unobserved elasticity matrix ${\Sigma}$:
\begin{align}
\mathbf{S}(\mathbf{p}, \mathbf{w} \mid {\Lambda}, {\Sigma}) \equiv 
\mathbf{S}(\mathbf{p}, \mathbf{w} \mid \mathbf{A}, {\Sigma}).
\end{align}


\subsection{The Dynamic Calibration (DC) Problem}
\label{subsec:dynamic_calibration}

The ultimate objective of processing networks is to predict how the supply-chain snapshot $\mathbf{S}$ transforms over time $t = 1, \dots, T$ in response to price fluctuations. Let $\mathbf{g}_t = ({g}_1(t), \dots, {g}_n(t))^\prime$ denote the vector of gross outputs, $\mathbf{d}_t = ({d}_1(t), \dots, {d}_n(t))^\prime$ the final demand, and $v_t \, (= \sum_{i=1}^n d_i(t))$ the total value added at time $t$, all measured in current value terms. 
From the fundamental material balance condition of the general equilibrium, the economy must satisfy the following accounting identity at any period $t$:
\begin{align}
\mathbf{S}(\mathbf{p}_t, \mathbf{w}_t \mid \mathbf{A}, {\Sigma}) \mathbf{g}_t = \mathbf{h}_t,
\label{eq:balance_equation}
\end{align}
where $\mathbf{h}_t = (\mathbf{g}_t - \mathbf{d}_t, v_t)^\prime \in \mathbb{R}^{n+1}$ represents the column vector collecting the intermediate demand sums and sectoral value added. For notational convenience, we define the model-predicted supply-chain mapping as $\hat{\mathbf{h}}_t({\Sigma}) \equiv \mathbf{S}\left( {\Sigma} \mid \mathbf{p}_t, \mathbf{w}_t, \mathbf{A} \right) \mathbf{g}_t$.

The Dynamic Calibration (DC) problem seeks to identify the fundamental elasticity matrix $\Sigma \in \mathbb{R}_+^{n \times n}$ that minimizes the discrepancy between the model-predicted transactions $\hat{\mathbf{h}}_t({\Sigma})$ and the historically observed realities $\mathbf{h}_t$ across the entire calibration window $t = 1, \dots, T$. As formally proven in Appendix \ref{AppendixA}, the non-negativity constraint ($\sigma_{ij} \ge 0$) is the \textit{necessary and sufficient} condition for the global concavity of the CCES unit cost functions, thereby guaranteeing the microeconomic regularity of the network. 

Formulating the objective function for economywide networks requires navigating a critical trade-off between macroeconomic aggregate consistency and microeconomic structural accuracy. A naive absolute-error least squares approach overwhelmingly prioritizes massive industries, severely distorting the substitution topology of smaller sectors. Conversely, relying solely on relative errors risks violating the aggregate macroeconomic accounting identities (e.g., the three-way equivalence of GDP). Furthermore, the highly non-linear system is intrinsically ill-posed when applied to noisy real-world data. To simultaneously resolve these issues, we formulate DC as a regularized dual-objective inverse optimization problem:
\begin{align}
\min_{\Sigma \ge 0}~ \mathscr{L}(\Sigma) 
= 
\sum_{t=1}^T \underbrace{\sum_{k=1}^{n+1} \left( \frac{\hat{h}_{k,t}({\Sigma}) - h_{k,t}}{\bar{h}_k} \right)^2}_{\text{Micro-structural relative error}} 
+ \alpha \sum_{t=1}^T \underbrace{\left( \frac{\mathbf{1}^\prime \hat{\mathbf{h}}_t({\Sigma}) - \mathbf{1}^\prime \mathbf{h}_t}{\mathbf{1}^\prime \bar{\mathbf{h}}} \right)^2}_{\text{Macro-aggregate GDP error}} 
+ \beta \underbrace{\vphantom{\left(\frac{\hat{h}}{\bar{h}}\right)^2} \| \Sigma - \mathbf{1}_{n \times n} \|^2}_{\text{Tikhonov regulator}}
\label{eq:inverse_optimization}
\end{align}
where $\bar{\mathbf{h}}$ is the time-averaged reality vector, $\mathbf{1}$ is a vector/matrix of ones, and $\alpha > 0$ is the weight enforcing macroeconomic aggregate consistency. The parameter $\beta > 0$ introduces a mild Tikhonov penalty that anchors the system to the Cobb-Douglas benchmark ($\sigma_{ij}=1$). This regularization ensures that the Allen-Uzawa Elasticity of Substitution (AUES) defaults to unity in the absence of strong empirical evidence to the contrary, effectively preventing economically implausible parameter explosions in an otherwise underdetermined subspace.

While the formulation in \eqref{eq:inverse_optimization} integrates micro and macro constraints, its computational execution is deceptively challenging for highly disaggregated processing networks (e.g., $n \ge 70$, involving thousands of parameters). Attempting to optimize $\mathscr{L}(\Sigma)$ as a monolithic block triggers a severe curse of dimensionality; evaluating the numerical Jacobian requires $\mathcal{O}(n^2)$ function calls per iteration, causing computational time to explode as $\mathcal{O}(n^4)$. Moreover, the gradient with respect to any $\sigma_{ij}$ involves highly non-convex terms due to the multi-layered recursive definitions embedded in the CCES structure. A change in a downstream elasticity non-linearly interferes with the gradient of an upstream elasticity, creating an extremely ill-conditioned, folded error manifold. When pushed against the non-negativity boundary ($\sigma_{ij} \ge 0$), generic gradient-based algorithms (e.g., L-BFGS-B) inevitably fall into active-set traps and prematurely halt at local minima with significant residuals.



\section{Algorithm for Dynamic Calibration}
\label{sec:methodology}

While standard econometric regressions (e.g., seemingly unrelated regressions on observed cost share matrices $\mathbf{A}_t$) are conventionally employed to estimate substitution elasticities, they face fatal limitations in macroeconomic processing networks. Estimating deep structural parameters using short-term panels frequently yields negative elasticities (e.g., \cite{NakanoNishimura2021}), severely violating the global concavity required for microeconomic regularity. Conversely, extending the time-series panel to gain statistical degrees of freedom relies on the empirically unrealistic assumption that fundamental substitution topologies remain invariant over decades. 

Given these parameter stability concerns, Dynamic Calibration (DC) emerges as a highly legitimate alternative. Crucially, rather than directly fitting the highly noisy, ex-post cost share actuals $\mathbf{A}_t$, DC is formulated to minimize discrepancies against the macroeconomic intermediate demand realities $\mathbf{h}_t$. By strictly enforcing economic regularity ($\Sigma \ge 0$) over a stable short-term window, DC acts as a robust constrained inverse optimization approach.

\subsection{The Failure of Monolithic Optimization}
To solve the regularized dual-objective problem defined in Eq. \eqref{eq:inverse_optimization}, one might naturally consider applying standard gradient-based algorithms (e.g., L-BFGS-B or Sequential Quadratic Programming). However, while such monolithic solvers occasionally succeed in small-scale pedagogical models (e.g., $n=10$), their computational execution scales disastrously, making them prohibitive for highly disaggregated, economywide networks (e.g., $n \ge 71$). Attempting to optimize all parameters simultaneously as a single block triggers three intrinsic mathematical complexities:
\begin{enumerate}
    \item \textsl{Curse of Dimensionality in Gradient Evaluation:} Monolithic optimization of an $n$-sector network requires navigating a parameter space of up to $\mathcal{O}(n^2)$ dimensions. Evaluating the numerical Jacobian at each iteration demands $\mathcal{O}(n^2)$ function evaluations, causing the computational burden to explode as $\mathcal{O}(n^4)$. In practical terms, calculating a single gradient step for an economywide network demands thousands of evaluations, paralyzing the solver.
    \item \textsl{Extreme Ill-conditioning of the CCES Hierarchy:} The multi-layered nesting of the CCES unit cost functions creates severe ill-conditioning. A marginal perturbation in an upstream elasticity parameter ($\sigma_n$) is exponentially amplified or dampened through the cascaded power functions before affecting the final cost shares. This extreme distortion causes the numerical Hessian matrix to degenerate, leading standard solvers to falsely declare convergence in flat plateaus.
    \item \textsl{The Active-Set Boundary Trap:} Searching exclusively within the economically meaningful domain ($\sigma_{ij} \ge 0$) frequently forces the monolithic solver against boundary constraints. Once trapped on these boundaries within a high-dimensional and non-convex folded manifold, generic algorithms struggle to escape, prematurely halting with significant residuals.
\end{enumerate}
Because of these interacting barriers, general-purpose solvers typically exhaust their evaluation limits before achieving meaningful error reduction. 

\subsection{The Hybrid Heuristic: Exploiting Network Topology}
To overcome this computational bottleneck, we introduce a deterministic \textit{structure-exploiting algorithm} (hybrid heuristic). Instead of treating the highly non-convex objective function $\mathscr{L}(\Sigma)$ as an $\mathcal{O}(n^2)$-dimensional black box, our algorithm leverages the physical topology of the processing network to structurally decompose the labyrinth into a sequence of low-dimensional, well-behaved subproblems. 

Our algorithm systematically alternates between two structural optimization modes:
\begin{itemize}
    \item \textsl{Vertical Cascade-Sequential Descent:} The algorithm explicitly exploits the hierarchical upstreamness of the network. Rather than optimizing the entire elasticity vector for industry $j$ simultaneously, the algorithm sequentially descends the cascade. It isolates and updates the parameters starting from the uppermost level. By evaluating the unit cost function recursively from top to bottom, the extreme ill-conditioning caused by nested interferences is systematically neutralized.
    \item \textsl{Horizontal Alternating Block Coordinate Descent:} While the vertical descent operates within a single industry $j$, the entire network is bound together horizontally by the macro-aggregate GDP constraint. Under the dual-objective formulation, a parameter update in one industry inevitably alters the economywide predicted transaction sum ($\mathbf{1}^\prime \hat{\mathbf{h}}_t$), thereby shifting the macro-error penalty for all other industries. To seamlessly handle this inter-sectoral interference without invoking monolithic evaluation, we employ an Alternating Block Coordinate Descent approach, updating the elasticity vector for one sector at a time.
\end{itemize}

\subsection{Algorithmic Implementation}
Let $k$ denote the global iteration (epoch) counter. Note that during DC, historical prices ($\mathbf{p}_t, \mathbf{w}_t$) are treated as exogenous observed data. The precise algorithmic flow to minimize $\mathscr{L}(\Sigma)$ is implemented as follows:

\vspace{0.2cm}
\noindent \textsl{Step 0: Initialization.} \\
Initialize the elasticity matrix at the Cobb-Douglas benchmark, $\Sigma^{(0)} = \mathbf{1}$. Set the epoch counter $k = 1$. Determine the hierarchical cascade sequence for each industry $j$ by sorting the inputs based on the upstreamness metric $(\mathbf{I}-\mathbf{A})^{-2}\mathbf{1}$.

\vspace{0.2cm}
\noindent \textsl{Step 1: Horizontal Alternating Block (Inter-Industry Loop).} \\
For each industry $j \in \{1, \dots, n\}$, we isolate its elasticity vector $\sigma_{\cdot j} \in \mathbb{R}_+^n$ while keeping the elasticities of all other industries ($\Sigma_{-j}$) fixed at their most recently updated values. The subproblem for sector $j$ is to minimize the partial objective:
\begin{align}
\min_{\sigma_{\cdot j} \ge 0} \mathscr{L} \left( \sigma_{\cdot j}, \Sigma_{-j} \right).
\end{align}
Instead of solving this subproblem simultaneously, we pass it to the vertical cascade descent.

\vspace{0.2cm}
\noindent \textsl{Step 2: Vertical Cascade-Sequential Descent (Intra-Industry Loop).} \\
Within industry $j$, let $i$ index the inputs ordered from the uppermost tier down to the bottommost tier. We optimize the scalar parameter $\sigma_{ij}$ one at a time, strictly descending the cascade:
\begin{align}
\sigma_{ij}^{(k)} = \arg\min_{\sigma_{ij} \ge 0} \mathscr{L} \left( \sigma_{ij} \mid \sigma_{<i, j}^{(k)}, \sigma_{>i, j}^{(k-1)}, \Sigma_{-j} \right),
\end{align}
where $\sigma_{<i, j}^{(k)}$ represents the parameters in the upstream layers already updated in the current epoch $k$, and $\sigma_{>i, j}^{(k-1)}$ represents the downstream parameters yet to be updated. This 1-dimensional bounded scalar optimization is solved efficiently using the L-BFGS-B algorithm.

\vspace{0.2cm}
\noindent \textsl{Step 3: Macro-Aggregate State Update.} \\
Once the vertical descent completes for industry $j$, its cost shares are accurately predicted by recursively calculating the CCES intermediary unit costs $\pi_i$ based strictly on the exogenous historical prices $\mathbf{p}_t$ (following Eq. \eqref{eq:cces_recursive}). The newly predicted transaction vector for sector $j$ is immediately reflected in the macroscopic environment by updating the economywide intermediate demand sum $\mathbf{1}^\prime \hat{\mathbf{h}}_t$. This provides the updated macro-aggregate penalty environment for the subsequent sector $j+1$.

\vspace{0.2cm}
\noindent \textsl{Step 4: Convergence Check.} \\
After completing the horizontal loop across all sectors $j = 1, \dots, n$, we evaluate the global error reduction:
\begin{align}
\Delta \mathscr{L} = \mathscr{L} \left( \Sigma^{(k)} \right) - \mathscr{L} \left( \Sigma^{(k-1)} \right).
\end{align}
If $|\Delta \mathscr{L}|$ falls below a predefined tolerance threshold $\epsilon$, the algorithm terminates and outputs the calibrated elasticity matrix $\Sigma^* = \Sigma^{(k)}$. Otherwise, set $k \leftarrow k + 1$ and return to Step 1.

\vspace{0.2cm}
By restricting the L-BFGS-B solver to 1-dimensional bounded optimizations (Step 2) nested within a Block Coordinate framework (Step 1), the hybrid heuristic breaks the $\mathcal{O}(n^4)$ monolithic bottleneck. This structure-exploiting design guarantees ultra-fast and robust convergence, enabling the dynamic calibration of highly disaggregated macro-networks.


\section{Results and Application}
\label{sec:results}

\subsection{The Concavity Constraint and Computational Performance}
We first apply the proposed Dynamic Calibration (DC) framework to a highly aggregated 10-sector model of the US economy. We define the calibration window as a stable short-term period ($T=4$ years around the base year) where the fundamental structural parameters are assumed to be invariant. Initially, estimating 100 elasticity parameters to satisfy 40 historical accounting identities appears to be an underdetermined problem. We anticipated an abundance of exact zero-residual solutions and intended to use the Tikhonov soft constraint to select the solution closest to neutrality (the Cobb-Douglas benchmark).

However, empirical implementation reveals a profound reality: when strictly imposing the necessary and sufficient condition for global concavity ($\Sigma \ge 0$), an exact zero-residual solution ceases to exist for noisy real-world data. A purely generic flexible function with $\mathcal{O}(n^2)$ parameters (e.g., Translog) might mathematically solve the system with zero residuals, but enforcing global concavity (i.e., negative semidefinite Hessian) on such functions during optimization is computationally intractable. The CCES processing network uniquely guarantees this absolute economic mandate via simple non-negativity bounds. 

Consequently, the strict economic discipline of concavity reveals that perfectly fitting the data is impossible. Yet, the identified non-zero minimum remains rather meaningful in an economic sense. By navigating the concave domain and anchoring toward neutrality (as justified in Section \ref{subsec:dynamic_calibration}), the resulting elasticity matrix $\Sigma^*$ represents the most theoretically sound substitution topology available. When tested on the 10-sector model using the dual-objective function, both the monolithic L-BFGS-B solver and our hybrid heuristic converged to this exact same theoretically optimal boundary (a mean absolute percentage error of 1.65\%, with a 0.00\% macroeconomic aggregate error). 


While monolithic solvers can reach this limit in small pedagogical models, they fail catastrophically when scaled. In our 71-sector economy-wide network (optimizing approximately 4,400 parameters), the monolithic L-BFGS-B solver explicitly suffers from the curse of dimensionality; it requires over 1,800 seconds (and 17,500 function evaluations) just to complete two initial iterations, stagnating at a high objective error of 6.38. In stark contrast, our structure-exploiting hybrid heuristic flawlessly circumvents this computational bottleneck, achieving a highly accurate, macro-consistent convergence (an error of 6.16) in roughly 1,000 seconds.

\subsection{Validation: Replicating T\"{o}rnqvist TFP via Structural Parameters}
\begin{figure}[t!]
    \centering
    \includegraphics[width=\textwidth]{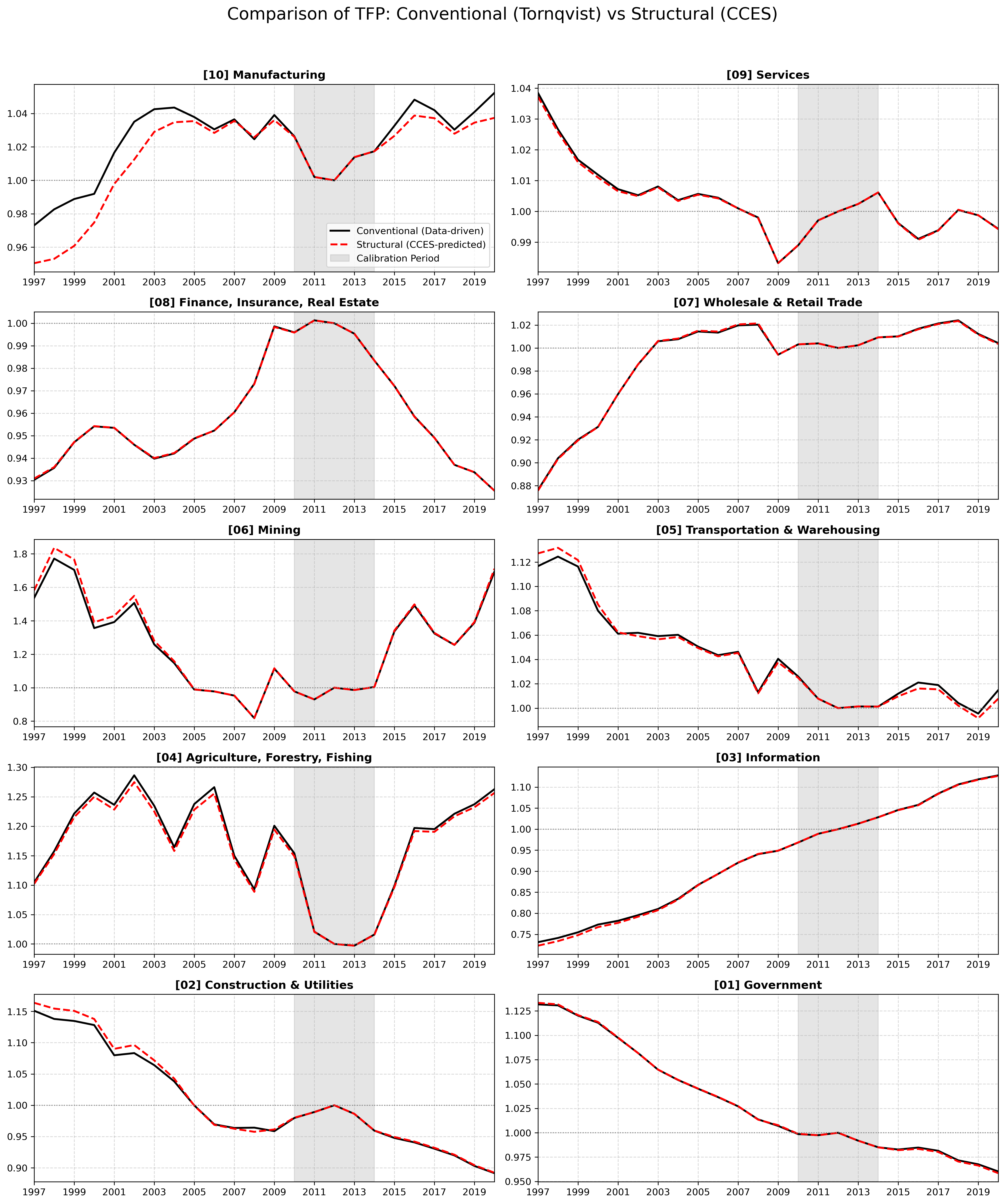}
    \caption{Comparison of TFP: T\"{o}rnqvist vs. CCES, ordered by upstreamness hierarchy.}
    \label{fig:tfp_comparison}
\end{figure}

Because empirical noise prevents a zero-residual exact fit, we must evaluate the validity of the calibrated elasticity matrix $\Sigma^*$ through its practical economic utility. Historically, a primary motivation for compiling time-series Input-Output tables is the ex-post measurement of macroeconomic productivity, notably the T\"{o}rnqvist Total Factor Productivity (TFP). As detailed in Appendix \ref{AppendixB}, the T\"{o}rnqvist index is the exact index for the Translog functional form; it calculates TFP directly from short-term fluctuations in ex-post cost shares. 

If our calibrated CCES model can parametrically replicate this data-driven metric without relying on year-to-year ex-post IO tables, it proves that the CCES matrix intrinsically captures the identical technological substitution potential of the Translog function. Using the calibrated structural parameters, we forward-simulate the cost share evolution (i.e., the endogenous supply-chain transformation) and extract the CCES-predicted structural TFP. 

Figure \ref{fig:tfp_comparison} superimposes this CCES structural TFP against the conventional T\"{o}rnqvist TFP. During the calibration window, the structural TFP mirrors the conventional indices with exceptional accuracy across diverse sectors. This alignment provides definitive evidence that the calibrated parameters are in effect legitimate, proving that our parametrically rigid, globally concave CCES model substantively captures the true technological dynamics of the network.


\subsection{Application: Endogenous Tail Risks and Mean Sinking}

Having validated the structural parameters, we apply the dynamically calibrated (nonneutroelastic) processing network to investigate the origins of macroeconomic tail risks and the consequent structural drag on aggregate productivity. Recall from Section \ref{subsec:cces_cost_shares} that the general equilibrium price system is defined as $\mathbf{p} = \left< \mathbf{z} \right>^{-1} \mathcal{C}(\mathbf{p}, \mathbf{w} \mid {\Lambda}, {\Sigma})$. Because the CCES unit cost function is strictly increasing and globally concave under the economic regularity condition $\Sigma \ge 0$, this continuous vector mapping $\mathbf{p} \mapsto \mathbf{p}$ guarantees global convergence to a unique positive fixed point, provided that a feasible fixed point exists \citep{Krasnoselskii1964, Kennan2001}. 

Leveraging this robust convergence property, we inject 50,000 samples of symmetric, normally distributed microeconomic productivity shocks ($\ln z \sim \mathcal{N}(0, 0.20^2)$) into the calibrated network and trace the distribution of the resulting aggregate output fluctuations ($\ln H$) via Monte Carlo simulations. This computational approach is crucial; while theoretical studies \citep{BaqaeeFarhi2019_2, NakanoNishimura2026} have suggested that network non-linearities (e.g., negative synergism) can induce tail asymmetry, establishing whether the actual economy is strictly \textit{inelastic}---and quantifying its systemic macroeconomic impact---requires structurally estimating the microeconomic elasticities ($\Sigma$) and executing the non-linear Domar aggregation at a general equilibrium.

Figure \ref{fig:qq_plot} presents the histograms and Quantile-Quantile (QQ) plots of the aggregate output fluctuations. The left panels demonstrate that a hypothetical, perfectly elastic Cobb-Douglas economy (where $\sigma_{ij} = 1$) translates normal micro-shocks linearly into normal aggregate growth (skewness: $-0.0108$). In stark contrast, the right panels display the response of the actual calibrated CCES economy, revealing a distinct, negatively skewed asymmetric heavy tail (skewness: $-0.3556$). 

\begin{figure}[t!]
    \centering
    \includegraphics[width=\textwidth]{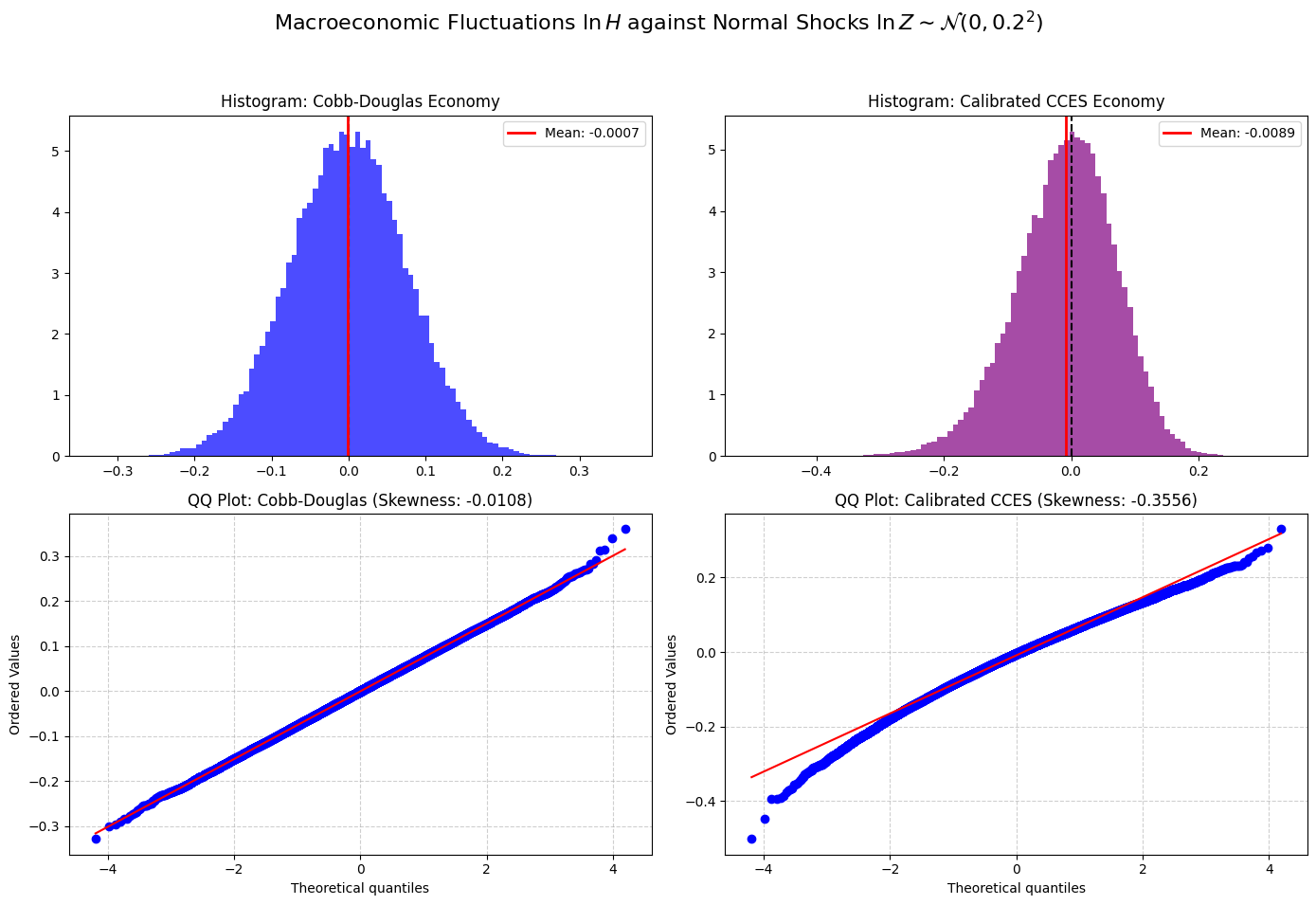}
    \caption{Macroeconomic Fluctuations against Normal Productivity Shocks: The linear Cobb-Douglas baseline (left) versus the empirically calibrated non-linear CCES network (right) revealing systemic tail risks and mean sinking.}
    \label{fig:qq_plot}
\end{figure}

Crucially, this asymmetric left tail does not merely represent an increase in variance (risk). As summarized in Table \ref{tab:mean_sinking_test}, the structural inelasticity of the network induces a systemic \textit{mean sinking} of aggregate output. Analogous to electrical \textit{impedance}, where an alternating current encounters resistance that permanently dissipates energy, the inelastic network acts as a macroeconomic impedance. It endogenously transforms symmetric microeconomic volatility (shocks with a mean of zero) into a permanent downward drag on expected aggregate productivity. A paired $t$-test strictly confirms that the expected aggregate productivity of the calibrated CCES network is statistically significantly lower ($p < 0.01$) than that of the linearly resilient Cobb-Douglas baseline.

\begin{table}[h!]
\centering
\caption{Moments of Aggregate Productivity Shocks ($\ln H$) and Paired $t$-test for Mean Sinking}
\label{tab:mean_sinking_test}
\begin{tabular}{l c c c}
\hline\hline
Model Specification & Mean & Variance & Skewness \\
\hline
Cobb-Douglas Economy ($\sigma_{ij}=1$) & $-0.00068$ & $0.00566$ & $-0.01079$ \\
Calibrated CCES Economy (Non-linear) & $-0.00885$ & $0.00612$ & $-0.35560$ \\
\hline
\multicolumn{4}{l}{\textit{Paired $t$-test (CCES vs. Cobb-Douglas)}} \\
Mean Difference (CCES $-$ CD) & \multicolumn{3}{l}{$-0.00817$} \\
$t$-statistic & \multicolumn{3}{l}{$-175.96$} \\
$p$-value & \multicolumn{3}{l}{$0.0000$} \\
\hline\hline
\multicolumn{4}{l}{\footnotesize Note: Results based on $N=50,000$ Monte Carlo simulations with $\ln Z \sim \mathcal{N}(0, 0.20^2)$.}
\end{tabular}
\end{table}

This empirical confirmation of the mean sinking effect structurally corroborates the premise that the US economy around the calibration window was inherently inelastic. An inelastic network lacks the requisite substitutability to absorb supply-chain bottlenecks. As theoretically explored in \cite{NakanoNishimura2026}, an inelastic processing network subjected to extreme negative productivity shocks is vulnerable to a Negative Singularity---a catastrophic structural breakdown where a positive equilibrium fixed point fails to exist. 

To be clear, the empirical inelasticity of the US economy during this specific window does not directly indicate an imminent singularity. Rather, it illustrates a fundamental operational limit: if a rigid, inelastic economy were forced to undergo a massive, systemic introduction of productivity-lowering constraints (such as stringent climate-change regulations or sudden geopolitical trade barriers), it could eventually breach its metabolic capacity and reach this structural breaking point. While quantitatively mapping the precise threshold of this singularity remains a vital subject for future research, the asymmetric tail risks and the statistically significant mean sinking effect observed here underscore the critical role of non-linear substitution elasticities in evaluating the true resilience and structural productivity of an economy.
\section{Concluding Remarks}
\label{sec:conclusion}
In this paper, we introduced the Cascaded CES (CCES) general equilibrium model and developed the Dynamic Calibration (DC) algorithm to provide an empirical framework for evaluating the endogenous structural transformation of economywide production networks. Moving beyond the limitations of static, neutral network models, our framework explicitly incorporates the hierarchical engineering structure of actual material flows and non-unitary substitution elasticities.

By applying the DC algorithm to U.S. time-series data, we effectively overcame the curse of dimensionality inherent in highly non-linear, multi-sector nested structures. This methodological advancement enabled us to organically reproduce the T\"ornqvist TFP index as an endogenous structural parameter. More crucially, it allowed us to empirically demonstrate the emergence of asymmetric macroeconomic tail risks driven by the network's inherent rigidities. We revealed that such structural inelasticity acts as a macroeconomic \textit{impedance}, translating symmetric micro-shocks into a statistically significant \textit{mean sinking} effect that structurally depresses aggregate productivity.

While our approach opens a new avenue for structural macroeconomic analysis, we acknowledge certain data-driven limitations in the current implementation. Specifically, the value-added components in our empirical application were not disaggregated into capital and labor, reflecting the consolidated nature of the time-series input-output (IO) tables used. Furthermore, the aggregation into 10 broad sectors was deliberately chosen; the primary objective of this paper is to serve as a proof-of-concept, demonstrating the computational feasibility and practical validity of the CCES and DC framework. Overfitting to the statistical artifacts of mechanical IO balancing procedures (e.g., the RAS method) is less meaningful than capturing the broader dynamic trends of endogenous structural change.

The tractability and efficiency of our framework pave the way for several highly practical policy applications and promising directions for future research. First, because the DC algorithm is remarkably data-efficient, relying only on sectoral output, price deflators, and final demand data, it can be scaled to much higher resolutions (e.g., hundreds of sectors). Expanding the empirical scope to include other major economies will facilitate cross-country benchmarking of structural robustness and economic dynamics.

Finally, the analytical insights derived from the CCES model can fundamentally reshape how we approach macroeconomic policy and structural vulnerability. As global economies transition toward stricter sustainability standards, future studies could utilize our framework to quantitatively map the precise thresholds of negative singularity---predicting systemic limits under the large-scale adoption of productivity-lowering constraints. Furthermore, analyzing how the permanent \textit{mean sinking} effect (network impedance) alters rational expectations and intertemporal welfare remains a vital theoretical frontier. We believe that integrating such engineering-inspired hierarchical structures into empirical macroeconomic modeling will remain a rich and vital frontier for understanding the true dynamics of the global economy.

\appendix
\normalsize
\renewcommand{\thesection}{\Alph{section}}
\makeatletter
\renewcommand{\@seccntformat}[1]{Appendix \csname the#1\endcsname\quad}
\makeatother
\numberwithin{equation}{section}

\section{Concavity Proof for CCES Functions} \label{AppendixA}

A fundamental microeconomic requirement for the CCES processing network is that the unit cost function must be globally concave with respect to the price vector. By the well-established duality between production and cost \citep[e.g.,][]{Diewert1974}, the concavity of the unit cost function with respect to prices is mathematically equivalent to the global concavity of the underlying production function with respect to input quantities. 

Therefore, to establish the necessary and sufficient conditions for the economic regularity of the CCES framework, we define the corresponding CCES production function using input quantities $x_i$, and formally prove its global concavity.

\begin{definition} 
The CCES production function corresponding to the cost structure in Eq. \eqref{eq:cces_recursive} is defined as follows:
\begin{align}
{\chi}_{i-1} = 
\left( ({\lambda}_i)^{\frac{1}{{\sigma}_i}} ({x}_i)^{\frac{{\sigma}_i - 1}{\sigma_i}} + (1 - {\lambda_i})^{\frac{1}{{\sigma}_i}} ({\chi}_i)^{\frac{{\sigma}_i - 1}{\sigma_i}} \right)^{\frac{\sigma_i}{{\sigma}_i - 1}}, &&  i = 1, \dots, n,
\label{cces_pf}
\end{align}
where ${\chi}_{0} = y$ is the final output quantity, and ${\chi}_{n} = \ell$ is the primary factor (value-added) quantity.
\end{definition}

\begin{proposition}
The CCES production function ${F}({x}_1, \dots, {x}_{n}, \ell)$ is globally concave if and only if ${\sigma}_{i} \geq 0$ for all ${i} = 1, \dots, {n}$. By duality, this is the necessary and sufficient condition for the CCES unit cost function to be globally concave with respect to prices.
\end{proposition}

\begin{lemma}
The CCES function ${F}({x}_1, \dots, {x}_{n}, \ell)$ is concave at the reference point $(1, \dots, 1)$ only if the substitution elasticities satisfy ${\sigma}_{k} \geq 0$ for all ${k} = 1, \dots, {n}$.
\end{lemma}

\begin{proof}
Consider a specific hierarchical level ${k}$. We evaluate the condition by fixing all upstream variables ${x}_{{k}+1}, \dots, {x}_{n}$ and $\ell$ to 1. For notational brevity, we write ${\chi} = {\chi}_{{k}-1}$ and ${x} = {x}_{k}$.
The overall CCES function can then be decomposed as follows:
\begin{align}
    {\chi}_0 = {F}(\mathbf{x}, {\chi}), && {\chi} = {G}({x}, 1),
\end{align}
where $\mathbf{x} = ({x}_1, \dots, {x}_{{k}-1})$ represents the vector of downstream inputs, ${\chi} = {G}({x}, 1)$ is the composite output from level ${k}$, and ${F}$ is the aggregate downstream function. Note that ${F}$ is homogeneous of degree one with respect to its arguments $(\mathbf{x}, {\chi})$.

We evaluate the $k \times k$ local Hessian matrix $\mathbf{H}$ with respect to $(\mathbf{x}, {x})$ at the reference point $(1, \dots, 1)$. Using the chain rule, $\mathbf{H}$ can be partitioned into block matrices:
\begin{equation}
    \mathbf{H} = \begin{pmatrix} 
    {F}_{\mathbf{x}^\prime \mathbf{x}} & 
    {F}_{\mathbf{x}^\prime{\chi}}{G}_{x} \\ 
    {G}_{x} {F}_{{\chi}\mathbf{x}} & 
    {F}_{{x}{x}}
    \end{pmatrix}
\end{equation}
where ${F}_{\mathbf{x}^\prime \mathbf{x}} = \frac{\partial^2 {F}}{\partial {x}_{i}^\prime \partial {x}_{j}}$, ${F}_{\mathbf{x}{\chi}} = \frac{\partial^2 {F}}{\partial {x}_{i} \partial {\chi}}$, and ${F}_{{\chi}\mathbf{x}} = \frac{\partial^2 {F}}{\partial {\chi} \partial {x}_{j}}$, for $i, j = 1, \dots, k-1$.
Note also that:
\begin{align*}
{G}_{x} &= \frac{\partial {G}}{\partial {x}} 
= \frac{\partial {\chi}_{{k}-1}}{\partial {x}_{k}}
= ({\lambda}_k)^{\frac{1}{{\sigma}_k}}
\\
{F}_{xx} &= 
\frac{\partial^2 {F}}{\partial {\chi}^2} \left( \frac{\partial {G}}{\partial {x}} \right)^2 + \frac{\partial {F}}{\partial {\chi}}\frac{\partial^2 {G}}{\partial {x}^2}
= {F}_{{\chi} {\chi}} ({\lambda}_{k})^{\frac{2}{{\sigma}_k}} + {F}_{\chi} {G}_{{x}{x}}.
\end{align*}

Because ${F}$ is homogeneous of degree one, Euler's theorem dictates that the Hessian of ${F}$ multiplied by the reference point vector $(\mathbf{1}, 1)^\prime$ yields the zero vector:
\begin{equation*}
    \begin{pmatrix} 
    {F}_{\mathbf{x}^\prime \mathbf{x}} & {F}_{\mathbf{x}^\prime {\chi}} \\ 
    {F}_{{\chi} \mathbf{x}} & {F}_{{\chi} {\chi}} 
    \end{pmatrix} 
    \begin{pmatrix} \mathbf{1}^\prime \\ 1 \end{pmatrix} 
    = \begin{pmatrix} \mathbf{0}^\prime \\ 0 \end{pmatrix}.
\end{equation*}
This provides the following identities: 
\begin{align}
{F}_{\mathbf{x}^\prime \mathbf{x}} \mathbf{1}^\prime = - {F}_{\mathbf{x}^\prime {\chi}},
&&
{F}_{{\chi} \mathbf{x}} \mathbf{1}^\prime = - {F}_{{\chi} {\chi}}.
\label{euler}
\end{align}

Consider the following direction vector:
\begin{align*}
    \mathbf{v} 
= \left( ({\lambda}_k)^{\frac{1}{{\sigma}_k}}, \dots, ({\lambda}_k)^{\frac{1}{{\sigma}_k}}, 1 \right) = \left( ({\lambda}_k)^{\frac{1}{{\sigma}_k}} \mathbf{1}, 1 \right).
\end{align*}
The second-order directional derivative along $\mathbf{v}$ is given by the quadratic form:
\begin{align*}
\mathbf{v} \mathbf{H} \mathbf{v}^\prime 
&= ({\lambda}_k)^{\frac{2}{{\sigma}_k}} \left(
\mathbf{1} {F}_{\mathbf{x}^\prime \mathbf{x}} \mathbf{1}^\prime  
+ {F}_{{\chi}\mathbf{x}} \mathbf{1}^\prime 
+ \mathbf{1} {F}_{\mathbf{x}^\prime {\chi}}  
+ {F}_{\chi \chi}  
\right)
+ {F}_{\chi} {G}_{{x}{x}} 
\\
&= {F}_{\chi} {G}_{{x}{x}},
\end{align*}
where we used the identities in \eqref{euler} to cancel the terms in the parenthesis.
For the CCES function, the derivatives at the reference point evaluate to:
\begin{align*}
{F}_{\chi} &= \frac{\partial {\chi}_0}{\partial {\chi}_{{k}-1}} 
= \frac{\partial {\chi}_0}{\partial {\chi}_1}\frac{\partial {\chi}_1}{\partial {\chi}_2} \cdots \frac{\partial {\chi}_{{k}-2}}{\partial {\chi}_{{k}-1}} 
= \prod_{{i}=1}^{{k}-1} (1 - {\lambda}_i)^{\frac{1}{{\sigma}_i}}, \\
{G}_{xx} &= 
\frac{\partial^2 {\chi}_{{k}-1}}{ \partial {x}_{k}^2}
= - \frac{({\lambda}_k)^{\frac{1}{{\sigma}_{k}}} (1 - {\lambda}_k)^{\frac{1}{{\sigma}_k}}}{{\sigma}_{k}}.
\end{align*}
Hence, we obtain the exact directional derivative:
\begin{align}
\mathbf{v} \mathbf{H} \mathbf{v}^\prime 
= - \frac{({\lambda}_k)^{\frac{1}{{\sigma}_{k}}} \prod_{{i}=1}^{k} (1 - {\lambda}_i)^{\frac{1}{{\sigma}_i}}}{{\sigma}_{k}}.
\end{align}
Because $\lambda_i \in (0, 1)$, the numerator is strictly positive. Therefore, if ${\sigma}_{k} < 0$, then $\mathbf{v} \mathbf{H} \mathbf{v}^\prime > 0$. This strictly positive curvature implies the existence of a direction $\mathbf{v}$ along which the function is strictly convex, directly contradicting the necessary condition for global concavity. Thus, for the CCES function to be globally concave, it is strictly necessary that ${\sigma}_{k} \geq 0$ for all ${k}$.
\end{proof}

\begin{lemma}
If ${\sigma}_{i} \geq 0$ for all ${i} = 1, \dots, {n}$, the CCES function ${F}({x}_1, \dots, {x}_{n}, \ell)$ is globally concave.
\end{lemma}

\begin{proof}
Focus on a specific level ${k}$. The composite CCES function from level ${k}$ and its upstream components, ${\chi}_{{k}-1} = {G}_{k} ( {x}_{k}, {x}_{{k}+1}, \dots, {x}_{n}, \ell ) $, can be structurally decomposed as:
\begin{align*}
{\chi} = {F}_{k}({x}, {m}), && {m} = {G}_{{k}+1}({\Xi}),
\end{align*}
where we define ${\chi} = {\chi}_{{k}-1}$, ${x} = {x}_{k}$, ${m} = {\chi}_{k}$, and ${\Xi} = ({x}_{{k}+1}, \dots, {x}_{n}, \ell)$ for brevity. The ${k}$-th two-factor aggregation is denoted by ${F}_{k}$, such that ${F}_{k}({x}, {G}_{{k}+1}({\Xi})) = {G}_{k}({x}, {\Xi})$. Note that standard two-factor CES aggregations like ${F}_{k}$ are strictly concave if and only if their elasticity parameter satisfies ${\sigma}_{k} \ge 0$. 

Assume, by induction, that the upstream composite function ${G}_{{k}+1}$ is concave. Because both ${F}_{k}$ and ${G}_{{k}+1}$ are concave functions, they satisfy:
\begin{align*}
{F}_{k}({\theta} {x}_{a} + (1-{\theta}) {x}_{b}, {\theta} {m}_{a} + (1-{\theta}) {m}_{b}) &\geq {\theta} {F}_{k}({x}_{a}, {m}_{a}) + (1- {\theta}){F}_{k}({x}_{b}, {m}_{b}), \\
{G}_{{k}+1}({\theta} {\Xi}_{a} + (1-{\theta}) {\Xi}_{b}) &\geq {\theta} {G}_{{k}+1}({\Xi}_{a}) + (1- {\theta}) {G}_{{k}+1}({\Xi}_{b}),
\end{align*}
for any combinations ${x}_{a}, {x}_{b}, {\Xi}_{a}, {\Xi}_{b}$, and any ${\theta} \in [0, 1]$.
Let ${x}_{\theta} = {\theta} {x}_{a} + (1-{\theta}) {x}_{b}$, and ${\Xi}_{\theta} = {\theta} {\Xi}_{a} + (1-{\theta}) {\Xi}_{b}$. Since the outer function ${F}_{k}$ is monotonically increasing with respect to the upstream composite ${m}$ and jointly concave in $({x}, {m})$, it follows that:
\begin{align*}
{G}_{k}({x}_{\theta}, {\Xi}_{\theta}) 
&=  {F}_{k}({x}_{\theta}, {G}_{{k}+1}({\Xi}_{\theta})) \\
&\geq {F}_{k} \left( {x}_{\theta}, {\theta} {G}_{{k}+1}({\Xi}_{a}) + (1- {\theta}){G}_{{k}+1}({\Xi}_{b}) \right) \\
&= {F}_{k} \left( {\theta} {x}_{a} + (1-{\theta}) {x}_{b}, {\theta} {G}_{{k}+1}({\Xi}_{a}) + (1- {\theta}){G}_{{k}+1}({\Xi}_{b}) \right) \\
&\geq {\theta} {F}_{k}({x}_{a}, {G}_{{k}+1}({\Xi}_{a})) + (1-{\theta}) {F}_{k}({x}_{b}, {G}_{{k}+1}({\Xi}_{b})),
\end{align*}
which directly implies:
\begin{align*}
{G}_{k} \left( {\theta} {x}_{a} + (1-{\theta}) {x}_{b}, {\theta} {\Xi}_{a} + (1-{\theta}) {\Xi}_{b} \right) 
\geq {\theta} {G}_{k}({x}_{a}, {\Xi}_{a}) + (1-{\theta}) {G}_{{k}}({x}_{b}, {\Xi}_{b}).
\end{align*}
That is, the global concavity of ${G}_{k}$ is strictly implied by the concavity of its upstream counterpart ${G}_{{k}+1}$ and the condition ${\sigma}_{k} \geq 0$. Since the terminal upstream function ${G}_{n}$ is simply a two-factor concave function under ${\sigma}_{n} \geq 0$, the lemma holds universally by backward induction.
\end{proof}


\section{Data Construction and Sector Mapping} \label{AppendixB}

\subsubsection*{Time-Series Input-Output Tables and Deflators}
To construct the dynamically calibrated (DC) processing network, we extract temporal economic transactions and price deflators from time-series input-output (IO) tables. For each period $t$, statistical agencies provide both real and nominal IO tables, as illustrated in Table \ref{tab:two-tables}. 

\begin{table}[h!]
  \centering
  \begin{minipage}[t]{0.48\textwidth}
    \centering
    \begin{tabular}[t]{ccccc | c}
      ${x}_{11}$ & ${x}_{12}$ & $\cdots$ & ${x}_{1n}$ & ${f}_1$ & ${y}_1$ \\
      ${x}_{21}$ & ${x}_{22}$ & $\cdots$ & ${x}_{2n}$ & ${f}_2$ & ${y}_2$ \\
      $\vdots$   & $\vdots$   & $\ddots$ & $\vdots$ & $\vdots$ & $\vdots$ \\
      ${x}_{n1}$ & ${x}_{n2}$ & $\cdots$ & ${x}_{nn}$ & ${f}_n$ & ${y}_n$ \\
      ${\ell}_{1}$  & ${\ell}_{2}$  & $\cdots$ & ${\ell}_{n}$  & \multicolumn{1}{c}{} & \\[2pt]
      \cline{1-4} \noalign{\vspace{0.9pt}} \cline{1-4}
      ${y}_{1}$  & ${y}_{2}$  & $\cdots$ & $y_{n}$  & \multicolumn{1}{c}{} &
    \end{tabular}
  \end{minipage}
  \hfill
  \begin{minipage}[t]{0.48\textwidth}
    \centering
    \begin{tabular}[t]{ccccc | c}
      ${p}_{1} {x}_{11}$ & ${p}_{1} {x}_{12}$ & $\cdots$ & ${p}_{1} {x}_{1n}$ & ${p}_{1} {f}_{1}$ & ${p}_{1} {y}_{1}$ \\
      ${p}_{2} {x}_{21}$ & ${p}_{2} {x}_{22}$ & $\cdots$ & ${p}_{2} {x}_{2n}$ & ${p}_{2} {f}_{2}$ & ${p}_{2} {y}_{2}$ \\
      $\vdots$   & $\vdots$   & $\ddots$ & $\vdots$ & $\vdots$ & $\vdots$ \\
      ${p}_{n} {x}_{n1}$ & ${p}_{n} {x}_{n2}$ & $\cdots$ & ${p}_{n} {x}_{nn}$ & ${p}_{n} {f}_{n}$ & ${p}_{n} {y}_{n}$ \\
      ${w}_{1} {\ell}_{1}$  & ${w}_{2} {\ell}_{2}$  & $\cdots$ & ${w}_{n} {\ell}_{n}$  & \multicolumn{1}{c}{} & \\[2pt]
      \cline{1-4} 
      ${p}_{1} {y}_{1}$  & ${p}_2 {y}_{2}$  & $\cdots$ & ${p}_{n} {y}_{n}$  & \multicolumn{1}{c}{} &
    \end{tabular}
  \end{minipage}
  \caption{Input-output tables in real (left) and nominal (right) terms.}
  \label{tab:two-tables}
\end{table}

In the real IO table (left), the material flow strictly satisfies the row balance condition: $\sum_{j=1}^{n} {x}_{ij} + {f}_{i} = {y}_{i}$. However, except for the base year (e.g., 2012 for US data) where all prices are normalized to unity ($p_i = w_j = 1$), the column sum (accounting balance) of the real table does not hold; hence the double horizontal line. Conversely, the nominal IO table (right) strictly satisfies both the row balance and the column accounting balance: $\sum_{i=1}^{n} {p}_{i} {x}_{ij} + {w}_{j} {\ell}_{j} = {p}_{j} {y}_{j}$.

By exploiting this structure, the implicit price deflators for sectoral outputs and value-added inputs at time $t$ are derived simply by dividing the nominal entries by their corresponding real entries. For our DC algorithm, the target variables are extracted directly from the nominal table: the macroeconomic intermediate demand sums and value-added sums, effectively constructed from the final demand elements $d_i = p_i f_i$ and value-added elements $v_j = w_j \ell_j$. Furthermore, the baseline cost share parameters $a_{ij}$ used to anchor the processing network are directly obtained from the nominal table as $a_{ij} = (p_i x_{ij}) / (p_j y_j)$.

\subsubsection*{Measurement of T\"{o}rnqvist TFP}
The conventional T\"{o}rnqvist index provides an ex-post measure of Total Factor Productivity (TFP) and is theoretically renowned as the exact index for the Translog functional form. Using the dual (price-based) approach, the T\"{o}rnqvist TFP growth for sector $j$ evaluates productivity as the difference between the share-weighted growth of input prices and the growth of the output price. Let $a_{ij}(t)$ denote the nominal cost share of input $i$ in sector $j$ at time $t$ (including value-added components). The log-change in T\"{o}rnqvist TFP is calculated as:
\begin{align}
\Delta \ln \text{TFP}_j(t) = - \ln \frac{p_j(t)}{p_j(t-1)} 
+ \sum_{i} \frac{a_{ij}(t-1) + a_{ij}(t)}{2} \ln \frac{p_i(t)}{p_i(t-1)} .
\end{align}
In Section \ref{sec:results}, we demonstrate that the CCES processing network can systematically replicate this Translog-exact index using purely invariant structural parameters $\Sigma^*$ without requiring ex-post share updating, proving its substantive validity.

\subsubsection*{Upstreamness and Sector Mapping}
To structurally decompose the high-dimensional optimization, our algorithm requires a pre-defined hierarchical cascade for the processing network. We identify this hierarchy using the concept of \textit{Upstreamness} \citep{Dietzenbacher2007, Antras2012}. 
Let $\mathbf{A}$ be the baseline input coefficient matrix. The upstreamness vector is computed as:
\begin{align}
(\mathbf{I} - \mathbf{A})^{-2} \mathbf{1} = \mathbf{I}\mathbf{1} + 2\mathbf{A}\mathbf{1} + 3\mathbf{A}^2\mathbf{1} + 4\mathbf{A}^3\mathbf{1} + \cdots
\label{eq:upstreamness}
\end{align}
This formulation measures the degree of backward linkage by heavily weighting the Average Propagation Length (APL) of intermediate demands. 

\begin{table}[t!]
\centering
\caption{Mapping of 71 Disaggregated Sectors to 10 Broad Sectors (Ordered by Upstreamness)}
\label{tab:sector_mapping}
\scriptsize
\begin{tabular}{p{5.5cm} p{10cm}}
\toprule
\textbf{10 Broad Sectors} & \textbf{Original 71 Sectors} \\
\midrule
\textbf{[10] Manufacturing} & [15] Apparel and leather and allied products; [16] Furniture and related products; [17] Miscellaneous manufacturing; [24] Printing and related support activities; [30] Other transportation equipment; [34] Electrical equipment, appliances, and components; [35] Nonmetallic mineral products; [36] Textile mills and textile product mills; [39] Wood products; [46] Machinery; [49] Computer and electronic products; [50] Food and beverage and tobacco products; [51] Plastics and rubber products; [55] Paper products; [57] Motor vehicles, bodies and trailers, and parts; [60] Fabricated metal products; [64] Petroleum and coal products; [67] Primary metals; [71] Chemical products \\
\addlinespace
\textbf{[09] Services} & [03] Social assistance; [04] Hospitals; [05] Nursing and residential care facilities; [07] Amusements, gambling, and recreation industries; [08] Ambulatory health care services; [10] Educational services; [19] Accommodation; [27] Performing arts, spectator sports, museums, and related activities; [28] Waste management and remediation services; [40] Food services and drinking places; [43] Computer systems design and related services; [44] Legal services; [45] Other services, except government; [62] Management of companies and enterprises; [65] Administrative and support services; [69] Miscellaneous professional, scientific, and technical services \\
\addlinespace
\textbf{[08] Finance, Insurance, Real Estate} & [01] Housing; [13] Funds, trusts, and other financial vehicles; [53] Rental and leasing services and lessors of intangible assets; [56] Securities, commodity contracts, and investments; [61] Federal Reserve banks, credit intermediation, and related activities; [63] Insurance carriers and related activities; [66] Other real estate \\
\addlinespace
\textbf{[07] Wholesale \& Retail Trade} & [06] Food and beverage stores; [09] General merchandise stores; [11] Motor vehicle and parts dealers; [25] Other retail; [70] Wholesale trade \\
\addlinespace
\textbf{[06] Mining} & [20] Support activities for mining; [42] Mining, except oil and gas; [68] Oil and gas extraction \\
\addlinespace
\textbf{[05] Transportation \& Warehousing} & [12] Water transportation; [14] Transit and ground passenger transportation; [18] Pipeline transportation; [26] Air transportation; [31] Rail transportation; [33] Warehousing and storage; [47] Truck transportation; [52] Other transportation and support activities \\
\addlinespace
\textbf{[04] Agriculture, Forestry, Fishing} & [32] Forestry, fishing, and related activities; [54] Farms \\
\addlinespace
\textbf{[03] Information} & [21] Publishing industries, except internet (includes software); [22] Motion picture and sound recording industries; [38] Data processing, internet publishing, and other information services; [58] Broadcasting and telecommunications \\
\addlinespace
\textbf{[02] Construction \& Utilities} & [48] Construction; [59] Utilities \\
\addlinespace
\textbf{[01] Government} & [02] Federal general government (defense); [23] State and local general government; [29] Federal government enterprises; [37] Federal general government (nondefense); [41] State and local government enterprises \\
\bottomrule
\end{tabular}
\end{table}

We can conceptualize the processing network as a multi-story building. The highest floors represent the most upstream sectors (e.g., Manufacturing), where raw value-added enters the cascade. As products flow downward through the intermediate layers, they progressively accumulate downstream features until they reach the ground floor (e.g., Government or Services), exiting the building as final demand. Table \ref{tab:sector_mapping} presents the aggregation mapping from the 71 original BEA (Bureau of Economic Analysis) input-output sectors to our 10 broad economic sectors. The broad sectors are deliberately numbered and ordered from the top floor [10] (most upstream) down to the ground floor [01] (most downstream), reflecting the descending sequence of our hybrid calibration heuristic.


\clearpage
\subsubsection*{Data availability}
The replication package, including the dataset and Python scripts used in this study, is available in the Zenodo repository at \url{https://doi.org/10.5281/zenodo.22789532}.

\subsubsection*{Funding}
This research was funded by the Japan Science and Technology Agency (JST) Social Scenario Research Program toward a carbon-neutral society (grant number JPMJCN2302).

\subsubsection*{Declaration of generative AI and AI-assisted technologies in the manuscript preparation process}
During the preparation of this work the authors used Gemini in order to improve language and edit the manuscript's text. After using this tool/service, the authors reviewed and edited the content as needed and take full responsibility for the content of the published article.

\subsubsection*{CRediT authorship contribution statement}
Kazuhiko Nishimura: Conceptualization, Methodology, Formal analysis, Writing --- Original Draft.
Satoshi Nakano: Conceptualization, Methodology, Data Curation, Validation, Writing --- Review \& Editing.

\bibliographystyle{apalike}
\raggedright
\bibliography{bibfile}

@article{KimNakanoNishimura2017,
title = {Multifactor CES general equilibrium: Models and applications},
journal = {Economic Modelling},
volume = {63},
pages = {115-127},
year = {2017},
issn = {0264-9993},
doi = {https://doi.org/10.1016/j.econmod.2017.01.024},
memo = {https://www.sciencedirect.com/science/article/pii/S0264999316304758},
author = {Jiyoung Kim and Satoshi Nakano and Kazuhiko Nishimura}
}

@article{Kopytov2024,
author = {Alexandr Kopytov AND Bineet Mishra AND Kristoffer Nimark AND Mathieu Taschereau-Dumouchel},
title = {{Endogenous Production Networks under Supply Chain Uncertainty}},
journal = {Econometrica},
volume = {92},
number = {5},
pages = {1621-1659},
url = {https://doi.org/10.3982/ECTA20629},
memo = {https://onlinelibrary.wiley.com/doi/abs/10.3982/ECTA20629},
memo = {https://onlinelibrary.wiley.com/doi/pdf/10.3982/ECTA20629},
year = {2024}
}

@article{LiuTsyvinski2023,
    author = {Liu, Ernest and Tsyvinski, Aleh},
    title = {{A Dynamic Model of Input-Output Networks}},
    journal = {The Review of Economic Studies},
    volume = {91},
    number = {6},
    pages = {3608-3644},
    year = {2024},
    month = {11},
    issn = {0034-6527},
    memo = {10.1093/restud/rdae012},
    memo = {https://doi.org/10.1093/restud/rdae012},
    memo = {https://academic.oup.com/restud/article-pdf/91/6/3608/60441701/rdae012.pdf},
}

@article{Chenery1949,
  author  = {Chenery, Hollis B.},
  title   = {{Engineering Production Functions}},
  journal = {The Quarterly Journal of Economics},
  year    = {1949},
  volume  = {63},
  number  = {4},
  pages   = {507--531},
  month   = nov,
  url     = {https://doi.org/10.2307/1882136}
}

@article{Gabaix2011,
	Author = {Gabaix, Xavier},
	Doi = {10.3982/ECTA8769},
	Journal = {Econometrica},
	Memo = {https://onlinelibrary.wiley.com/doi/abs/10.3982/ECTA8769},
	Number = {3},
	Pages = {733-772},
	Title = {{The Granular Origins of Aggregate Fluctuations}},
	Volume = {79},
	Year = {2011},
	url = {https://dx.doi.org/10.3982/ECTA8769}}

@article{Antras2012,
  title={{Measuring the upstreamness of production and trade in the global economy}},
  author={Antr{\`a}s, Pol and Chor, Davin and Fally, Thibault and Hillberry, Russell},
  journal={American Economic Review},
  volume={102},
  number={1},
  pages={412--433},
  year={2012},
 url={https://doi.org/10.1257/aer.102.3.412},
  publisher={American Economic Association}
}

@article{Dietzenbacher2007,
  title={{Production chains in an interregional framework: Identification by means of average propagation lengths}},
  author={Dietzenbacher, Erik and Romero, Isidoro},
  journal={International Regional Science Review},
  volume={30},
  number={4},
  pages={362--383},
  year={2007},
 url={https://doi.org/10.1177/0160017607305366},
  publisher={Sage Publications Sage CA: Los Angeles, CA}
}

@book{Krasnoselskii1964,
	Author = {Mark A. Krasnosel'ski\u{\i}},
	Issn = {0073-8442},
	Publisher = {Groningen, P. Noordhoff},
	Title = {{Positive Solutions of Operator Equations}},
	Year = {1964}}

@article{Kennan2001,
	Author = {John Kennan},
	Doi = {10.1006/redy.2001.0133},
	Issn = {1094-2025},
	Journal = {Review of Economic Dynamics},
	Number = {4},
	Pages = {893 - 899},
	Title = {{Uniqueness of Positive Fixed Points for Increasing Concave Functions on {R}n: An Elementary Result}},
	Volume = {4},
	Year = {2001},
	url = {http://dx.doi.org/10.1006/redy.2001.0133}}

@incollection{Diewert1974,
  author    = {Diewert, W. Erwin},
  title     = {{Applications of Duality Theory}},
  booktitle = {Frontiers of Quantitative Economics},
  volume    = {II},
  editor    = {Intriligator, Michael D. and Kendrick, David A.},
  pages     = {106--176},
  publisher = {North-Holland},
  address   = {Amsterdam},
  year      = {1974}
}

@article{NakanoNishimura2026,
  author  = {Nakano, Satoshi and Nishimura, Kazuhiko},
  title   = {{Nonlinear {Domar} aggregation over transforming production networks}},
  journal = {Journal of Evolutionary Economics},
  year    = {2026},
  volume  = {36},
  number  = {2},
  pages   = {58},
  memo     = {10.1007/s00191-026-00975-4},
  url     = {https://doi.org/10.1007/s00191-026-00975-4},
  issn    = {1432-1386}
}

@article{NakanoNishimura2018,
	Author = {Satoshi Nakano and Kazuhiko Nishimura},
	memo = {http://doi.org/10.1016/j.physa.2018.08.110},
	Issn = {0378-4371},
	Journal = {Physica A: Statistical Mechanics and its Applications},
	Memo = {http://www.sciencedirect.com/science/article/pii/S0378437118310537},
	Pages = {986 - 999},
	Title = {{Structural propagation in a production network with restoring substitution elasticities}},
	Volume = {512},
	Year = {2018},
	url = {https://dx.doi.org/10.1016/j.physa.2018.08.110}}

@article{NakanoNishimura2021,
	Author = {Satoshi Nakano and Kazuhiko Nishimura},
	memo = {10.1016/j.jmacro.2020.103216},
	Issn = {0164-0704},
	Journal = {Journal of Macroeconomics},
	Memo = {https://www.sciencedirect.com/science/article/pii/S0164070420301427},
	Pages = {103216},
	Title = {{Productivity propagation with networks transformation}},
	Volume = {67},
	Year = {2021},
	url = {https://dx.doi.org/10.1016/j.jmacro.2020.103216}}

@article{Acemoglu2012,
	Author = {Acemoglu, Daron and Carvalho, Vasco M. and Ozdaglar, Asuman and Tahbaz-Salehi, Alireza},
	memo = {10.3982/ECTA9623},
	Journal = {Econometrica},
	Number = {5},
	Pages = {1977-2016},
	Title = {The Network Origins of Aggregate Fluctuations},
	Volume = {80},
	Year = {2012},
	url = {https://dx.doi.org/10.3982/ECTA9623}}

@article{BaqaeeFarhi2019_2, 
	Author = {Baqaee, David Rezza and Farhi, Emmanuel},
	memo = {10.3982/ECTA15202},
	Journal = {Econometrica},
	Number = {4},
	Pages = {1155-1203},
	Title = {{The Macroeconomic Impact of Microeconomic Shocks: Beyond Hulten's Theorem}},
	Volume = {87},
	Year = {2019},
	url = {https://dx.doi.org/10.3982/ECTA15202}}

@book{Dixon2013,
     editor = {Dixon, Peter B. and Jorgenson, Dale W.},
  title     = {Handbook of Computable General Equilibrium Modeling},
  volume    = {1A--1B},
  year      = {2013},
  publisher = {North-Holland}
}

@book{Jorgenson1987,
  author    = {Jorgenson, Dale W. and Gollop, Frank M. and Fraumeni, Barbara M.},
  title     = {{Productivity and U.S. Economic Growth}},
  series    = {Harvard Economic Studies},
  volume    = {159},
  year      = {1987},
  publisher = {Harvard University Press},
  address   = {Cambridge, MA}
}

@article{Timmer2015,
author = {Timmer, Marcel P. and Dietzenbacher, Erik and Los, Bart and Stehrer, Robert and de Vries, Gaaitzen J.},
title = {{An Illustrated User Guide to the World Input--Output Database: the Case of Global Automotive Production}},
journal = {Review of International Economics},
volume = {23},
number = {3},
pages = {575-605},
url = {https://doi.org/10.1111/roie.12178},
memo = {https://onlinelibrary.wiley.com/doi/abs/10.1111/roie.12178},
eprint = {https://onlinelibrary.wiley.com/doi/pdf/10.1111/roie.12178},
year = {2015}
}

@incollection{Dawkins2001,
title = {{Chapter 58 - Calibration}},
editor = {James J. Heckman and Edward Leamer},
series = {Handbook of Econometrics},
publisher = {Elsevier},
volume = {5},
pages = {3653-3703},
year = {2001},
issn = {1573-4412},
url = {https://doi.org/10.1016/S1573-4412(01)05011-5},
memo = {https://www.sciencedirect.com/science/article/pii/S1573441201050115},
author = {Christina Dawkins and T.N. Srinivasan and John Whalley}
}

\end{document}